\documentclass[12pt,reqno]{article}

\usepackage[usenames]{xcolor}
\usepackage{amssymb}
\usepackage{graphicx}
\usepackage{amscd}

\usepackage[colorlinks=true,
linkcolor=webgreen,
filecolor=webbrown,
citecolor=webgreen]{hyperref}

\definecolor{webgreen}{rgb}{0,.5,0}
\definecolor{webbrown}{rgb}{.6,0,0}

\usepackage{color}
\usepackage{fullpage}
\usepackage{float}

\usepackage{graphics,amsmath,amssymb}
\usepackage{amsthm}
\usepackage{amsfonts}
\usepackage{latexsym}
\usepackage{epsf}
\usepackage{tikz}

\title{Counting Subwords and Regular Languages}

\author{Charles J. Colbourn and Ryan E. Dougherty\\
Computing, Informatics, and Decision Systems Engineering\\
Arizona State University\\
P.O. Box 878809\\
Tempe, AZ 85287-8809\\
USA\\
{\tt ryan.dougherty@asu.edu}\\
{\tt Charles.Colbourn@asu.edu}\\[.1in]
\and
Thomas Finn Lidbetter and Jeffrey Shallit\\
School of Computer Science \\
University of Waterloo\\
Waterloo, ON  N2L3G1 \\
Canada\\
{\tt finn.lidbetter@uwaterloo.ca}\\
{\tt shallit@uwaterloo.ca}}
\begin{document}

\maketitle
\theoremstyle{plain}
\newtheorem{theorem}{Theorem}
\newtheorem{corollary}{Corollary}
\newtheorem{lemma}{Lemma}
\newtheorem{proposition}{Proposition}

\theoremstyle{definition}
\newtheorem{definition}{Definition}
\newtheorem{example}{Example}
\newtheorem{conjecture}{Conjecture}

\theoremstyle{remark}
\newtheorem{remark}{Remark}

\begin{abstract}
Let $x$ and $y$ be words.  We consider the languages
whose words $z$ are those for which the numbers of occurrences
of $x$ and $y$, as subwords of $z$,
are the same (resp., the number of
$x$'s is less than the number of $y$'s, resp., is 
less than or equal).   We give a necessary and sufficient
condition on $x$ and $y$ for these languages to be regular, and we show how to check this condition efficiently.
\end{abstract}

\section{Introduction}

A major theme in formal language theory is counting occurrences of letters
in words.  Let $\Sigma$ be a finite alphabet.
For a word $z \in \Sigma^*$ and a letter $a \in \Sigma$ we write
$|z|_a$ for the number of occurrences of $a$ in $z$.   A classic example of a language that is context-free but
not regular is
$ \{ z \in \{a,b\}^* \ : \ |z|_a = |z|_b  \}$; see, for example,
\cite[Exercise 3.1 (e), p.~71]{Hopcroft&Ullman:1979}.  The Parikh map (see, e.g., 
\cite{Parikh:1966}) is another example of this theme.

We can generalize the counting of {\it letter\/} occurrences to the
counting of {\it word\/} occurrences.
Let $w, y \in \Sigma^*$.
We say $y$ is a {\it subword\/}\footnote{Sometimes called ``factor'', especially in the European literature.} of $w$ if there exist
$x, z \in \Sigma^*$ such that $w = xyz$.  
Define $|w|_y$ to be the number of (possibly overlapping) occurrences
of $y$ in $w$.  Thus, for example, $|{\tt banana}|_{\tt ana} = 2$.

In this paper we study the languages
\begin{align*}
L_{x<y} &= \{ z \in \Sigma^* \ : \ |z|_x < |z|_y \},  \\
L_{x\leq y} &= \{ z \in \Sigma^* \ : \ |z|_x \leq |z|_y \},  \\
L_{x=y} &= \{ z \in \Sigma^* \ : \ |z|_x = |z|_y \}  
\end{align*}
and their complements.
The following is easy to see:

\begin{proposition}
For all words $x$ and $y$, the languages
$L_{x<y}$, $L_{x \leq y}$, $L_{x=y}$, and their complements $L_{x\geq y}$, 
$L_{x>y}$, $L_{x\not=y}$
are all deterministic context-free languages.
\end{proposition}

\begin{proof}
We prove this only for $L_{x=y}$, with the other cases
being analogous.
We construct a deterministic pushdown automaton $M$ that recognizes $L_{x=y}$ as follows:
its states record the last $\max(|x|,|y|)-1$ letters of the input seen so far. 
The stack of $M$ is used as a counter to maintain the absolute value of the difference 
between the number of $x$'s seen so far and the number of $y$'s
(a flag in the state records the sign of the difference).
We have $M$ accept its input if and only if this difference is $0$.
Since there is only one possible action for every triple of state, input symbol, and top-of-stack symbol, $M$ is deterministic (any ``invalid'' configurations transition to a dead state $d$).   
\end{proof}

While $L_{x=y}$ is always deterministic context-free, sometimes --- perhaps
surprisingly --- it can also be regular.  For example,
when the underlying alphabet $\Sigma$ is unary, then $L_{x=y}$ is always regular.
Less trivially, for $\Sigma = \{0,1\}$ it is an easy exercise to show that $L_{01=10}$ is regular, and is
recognized by the $5$-state DFA in Figure~\ref{a01}; however, $L_{01=10}$ is not regular when $\Sigma = \{0,1,2\}$.
On the other hand, $L_{0011=1100}$ is never regular, even when $\Sigma = \{0,1\}$.  

\begin{figure}
\centering
%\begin{center}
\includegraphics[width=3in]{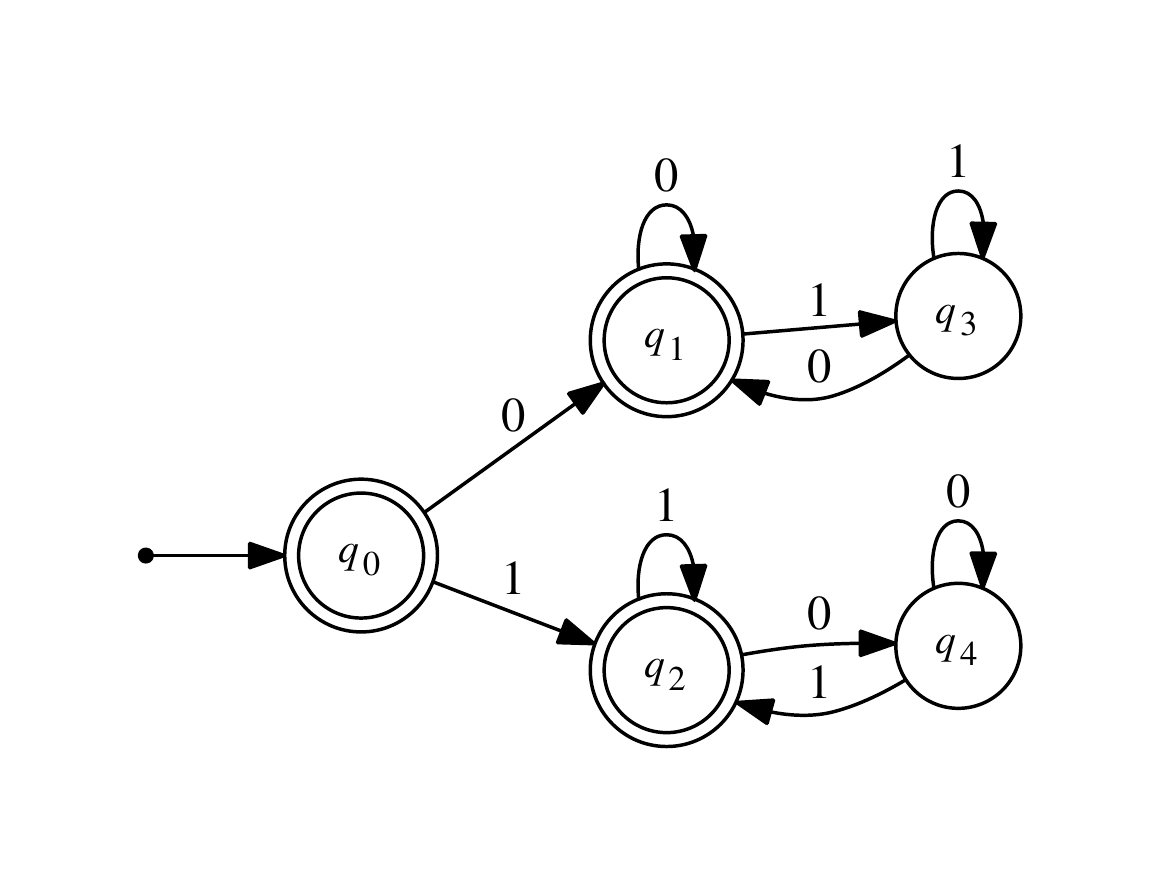}

\vskip -.35in
\caption{A DFA recognizing $L_{01=10}$ over 
$\Sigma = \{0,1\}$}
%\end{center}
\label{a01}
\end{figure}

The goal of this paper is to give a necessary and sufficient condition for these languages to be regular, and show how to check it efficiently.  The case
$L_{x<y}$ is covered in Section~\ref{lt1}, and the case
$L_{x=y}$ is covered in Section~\ref{eq1}.  The remaining case 
$L_{x \leq y}$ is virtually the same as $L_{x<y}$, and is left to the reader.

We assume a basic background in formal languages and automata; for
all unexplained notions, see \cite{Hopcroft&Ullman:1979}.
Four things are worth noting:  if $L$ is a language,
then we write
$L^c$ for the complement $\Sigma^* - L$.  Also, if $x$ is a word,
then $x^R$ denotes the reverse of the word $x$.  If $x = a_1 a_2 \cdots a_n$ is a word, with each $a_i \in \Sigma$, then $x[i..j] := a_i \cdots a_j$.  Finally,
if $x = tu$, then by $x u^{-1}$ we mean the word $t$, and by
$t^{-1} x$ we mean the word $u$.

\section{Bordered words and periodicity}

Let $y, z$ be words with $y$ nonempty.  We say that $z$ is {\it $y$-bordered} if
$z \not= y$ and $y$ is both a prefix and a suffix of $z$.
There are two types of $y$-bordered words:  one where
the prefix and suffix $y$ do not overlap in $z$ (that is,
where $|y| \leq |z|/2$), and one where they do 
(that is, where $|y| > |z|/2$).   In the first case, we say that $z$ is 
{\it disjoint\/} $y$-bordered, and in the second case,
{\it overlapping\/} $y$-bordered.
For example, {\tt entanglement} is disjoint {\tt ent}-bordered, and
{\tt alfalfa} is overlapping {\tt alfa}-bordered.
For more about borders of words, see, for example, \cite{Ehrenfeucht&Silberger:1979}. 

We will need two lemmas about bordered words.  

\begin{lemma}
Suppose $z \in \Sigma^*$ is
$y$-bordered.  Then there exist words $u \in \Sigma^+$ and
$v \in \Sigma^*$ and an integer $e \geq 0$
such that $y = (uv)^e u$ and $z = (uv)^{e+1}u$.
\label{ls}
\end{lemma}

\begin{proof}
Follows immediately from the Lyndon-Sch\"utzenberger theorem
(see, for example, \cite{Lyndon&Schutzenberger:1962} or
\cite[Theorem 2.3.2]{Shallit:2009}).   
\end{proof}

\begin{lemma}
Let $u \in \Sigma^+$,
$v \in \Sigma^*$, and $e \geq 0$.
Suppose that $y = (uv)^e u$.  Define $z_1 = (uv)^{e+1}$ and $z_2 = (uv)^{e+2}$.
Let $c = |z_1|_y$ and $d= |z_2|_y - |z_1|_y $.  Then
$c, d \geq 1$ and 
$|(uv)^i|_y = (i-e)d + c-d$ for all integers $i > e$.
\label{uvu}
\end{lemma}

\begin{proof}
If $x = (uv)^{e+1}$, then $|x|_y = c$.  Appending $uv$ to $x$ on the right
results in $d \geq 1$ additional copies of $y$.  The result now
follows by induction.
\end{proof}

We also recall the following classical result.

\begin{theorem}
Let $x,y$ be nonempty words.  There exists
a word with
two distinct factorizations as a concatenation
of $x$'s and $y$'s if and only if $xy = yx$.
\label{factor}
\end{theorem}

\begin{proof}
This follows from the so-called ``defect theorem"
\cite{Lothaire:1983}, or from \cite[Theorem 1]{Gamard&Richomme&Shallit&Smith:2017}.
\end{proof}

\section{Automata}
\label{autom}

We will need the following well-known result about pattern-matching automata (for example, see \cite[\S 32.3]{Cormen&Leiserson&Rivest&Stein:2001}).

\begin{theorem}
Given
a word $w \in \Sigma^n$, a DFA $M =
(\{ q_0, \ldots, q_n \}, \Sigma, \delta, q_0, \{ q_n \})$
exists of $n+1$ states such that $\delta(q_0, x) = q_n$
if and only if $w$ is a subword (resp., suffix) of $x$.   Here
the state $q_i$ can be interpreted as asserting that 
the longest suffix of the input that matches
a prefix of $w$ is of length $i$.  
\label{pma}
\end{theorem}

\section{Interlacing}

Suppose $y$ is a subword of every $x$-bordered word.  In this
case we say $x$ is {\it interlaced by} $y$.  For example,
it is easy to check that $000100$ is interlaced by $1000$ when
the underlying alphabet $\Sigma$ is $\{0,1\}$.   The following lemma gives
the fundamental property of interlacing:

\begin{lemma}
Suppose $x$ is interlaced by $y$, and suppose $z$ is a word satisfying
$|z|_y = |z|_x + k$.  Then for all $t$ we have
$|zt|_y \geq |zt|_x + k-1$.  In particular,
$|t|_y \geq |t|_x - 1$ for all $t$.
\label{count}
\end{lemma}

\begin{proof}
Identify the starting positions of all occurrences of $x$ in
$zt$.  Since $x$ is interlaced by $y$, between any two consecutive
occurrences of $x$, there must be at least one occurrence of $y$.
So if $zt$ has $i$ more occurrences of $x$ than $z$ does, then
$zt$ must have at least $i-1$ more occurrences of $y$ than
$z$ does.    
\end{proof}

\section{The language \texorpdfstring{$L_{x<y}$}{Lx<y}}
\label{lt1}

\begin{theorem}  The language $L_{x<y}$ is regular if and only if
either 
$x$ is interlaced by $y$ or
$y$ is interlaced by $x$.
\label{main1}
\end{theorem}

\begin{proof}
$\Longleftarrow$:  There are two cases: 
(i) $x$ is interlaced by $y$; and
(ii) $y$ is interlaced by $x$.

\bigskip

\noindent{\bf Case (i)}:
Using Lemma~\ref{count}, we can
build a finite automaton $M$ recognizing $L_{x<y}$ as follows:
using the pattern-matching automata for $x$ and $y$ described in Section~\ref{autom}, on input $z$ the machine $M$ records whether
\begin{itemize}
\item[(a)] $|z|_x= |z|_y + 1$;
\item[(b)] $|z|_x = |z|_y$; 
\item[(c)] $|z|_x = |z|_y - 1$, and
$|z'|_x \geq |z'|_y - 1$ for all prefixes
$z'$ of $z$;
\item[(d)] $|z'|_x \leq |z'|_y- 2$ for some
prefix $z'$ of $z$.
\end{itemize}
Of course we do not maintain the actual numbers
$|z|_x$ and $|z|_y$ in $M$, but only which of (a)--(d)
hold.
Lemma~\ref{count} implies that the four cases above cover all the possibilities. It is not possible to have $|z|_x\geq |z|_y+2$, and
if (d) ever occurs, we know from Lemma~\ref{count}
that $|z|_x < |z|_y$ for all words $z$ extending
$z'$.  So in this case the correct action is
for the automaton to remain in
state (d), an accepting state that loops to itself on all inputs.  The automaton accepts the input if and only if
it is in the states corresponding to conditions  (c) and (d).

\bigskip
\noindent{\bf Case (ii)}:
Using Lemma~\ref{count}, as in Case (i), we can
build a finite automaton recognizing $L_{x<y}$ as follows:
using the pattern-matching automata for $x$ and $y$ described in Section~\ref{autom}, on input $z$ the machine $M$ records whether
\begin{itemize}
\item[(a)] $|z|_y = |z|_x + 1$;
\item[(b)] $|z|_y = |z|_x $, and
$|z'|_y \geq |z'|_x $ for all prefixes
$z'$ of $z$;
\item[(c)] $|z'|_y \leq |z'|_x - 1$ for some
prefix $z'$ of $z$.
\end{itemize}
Lemma~\ref{count} implies that the three cases above cover all the possibilities. It is not possible to have $|z|_y\geq |z|_x+2$, and
if (c) ever occurs, we know from Lemma~\ref{count}
that $|z|_x  \geq |z|_y$ for all words $z$ extending
$z'$.  So in this case the correct action is
for the automaton to remain in
state (c), a rejecting ``dead'' state that loops to itself on all inputs.  The automaton accepts the input if and only if
it is in the state corresponding to condition (a).

\bigskip
\noindent $\Longrightarrow$:  
We proceed by proving the contrapositive. So suppose that there is some $y$-bordered word $r$ such
that $x$ is not a subword of $r$, and some
$x$-bordered word $s$ such that $y$ is not a subword
of $s$.
Using Lemma~\ref{ls}, we know that there are words
$u, v, p, q$ and natural numbers $e, f$ such that $r=(uv)^{e+1} u$, and $y=(uv)^e u$, and $s=(pq)^{f+1} p$, and $x=(pq)^f p$.

Suppose that $x$ is a subword of $(uv)^i u$ for some $i\geq 0$.
Since $x$ is not a subword of $r$, we know that 
$i \geq e+2$.  If $x$ is a subword of $(uv)^{e+2} u$ and not a subword
of $(uv)^{e+1} u$, then $y = (uv)^e u$ must be a subword of $x$.  
But then $y$ is a subword of $s$, a contradiction.  So $x$ is not
a subword of $(uv)^i u$ for any $i$.
By exactly the same reasoning we deduce that $y$ is not a subword
of $(pq)^j p$ for any $j\geq 0$.
 
Let $c = |(uv)^{e+1}|_y$ and $d = |(uv)^{e+2}|_y - |(uv)^{e+1}|_y$.
Similarly, define
$c' = |(pq)^{f+1}|_x$ and
$d' = |(pq)^{f+2}|_x - |(pq)^{f+1}|_x$.
Consider a word $z = (uv)^i (pq)^j $, where $i>e$ and $j>f$.  
From above and Lemma~\ref{uvu}, we know that 
$|(uv)^i|_x = 0$ for all $i\geq 0$ and $|(pq)^j|_x = 
(j-f)d' + c'-d'$ for $j > f$.
Let $m$ be the number of additional occurrences of $x$ that
straddle the boundary between $(uv)^{e+1}$ and $(pq)^{f+1}$. That is, $m$ is the number of distinct values for $k$, such that $x$ is a subword of $(uv)^{e+1}(pq)^{f+1}$ starting at index $k$ and $(e+1)|uv|+2-|x|\leq k\leq (e+1)|uv|+1$.
Similarly, we know that $|(uv)^i|_y = (i-e)d +
c-d$ for $i > e$ and $|(pq)^j|_y = 0$ for all $j\geq 0$.
Let $n$ be the number of additional occurrences of $y$ that
straddle the boundary between $(uv)^{e+1}$ and $(pq)^{f+1}$. The precise definition of $n$ is given as above by replacing $m$ and $x$ with $n$ and $y$ respectively.
Thus $z$ has $(j-f)d'+c'-d'+m$ occurrences of $x$ and
$(i-e)d +c-d +n $ occurrences of $y$.  

Now assume, contrary to what we want to prove,
that $L_{x<y}$ is regular.  Define $L = L_{x<y} \cap (uv)^e(uv)^+ (pq)^f(pq)^+ $.  Then $L$ is regular.
Define a morphism $h: \{a,b\}^* \rightarrow \Sigma^*$ as
follows:  $h(a) = uv$, and
$h(b) = pq$.  We claim that
$h^{-1}(  z )= \{  a^i b^j \}$.  One direction is clear.   For the other, suppose $h^{-1}(z)$
included some word other than $a^i b^j$.
Then by Theorem~\ref{factor}, we know that
$uv$ and $pq$ commute.  But then by
the Lyndon-Sch\"utzenberger theorem
\cite{Lyndon&Schutzenberger:1962}, $uv$ and
$pq$ are both powers of some word $t$.  But then
$x$ would be a subword of $(uv)^\ell u$
for some $\ell$, which we already saw to
be impossible.  

By a well-known theorem (e.g.,
\cite[Theorem 3.3.9]{Shallit:2009}), $h^{-1}(L)$
is regular.  But $h^{-1}(L) = 
\{ a^i b^j  \ :\ (i-e)d +c-d +n < (j-f)d'+c'-d'+m, \ \text{for} \ i>e, \ j>f  \}$ which,
using the pumping lemma, is not regular.  
\end{proof}

\section{The language \texorpdfstring{$L_{x=y}$}{Lx=y}}
\label{eq1}

\begin{theorem}  The language $L_{x=y}$ is regular if and only if
either 
$x$ is interlaced by $y$ or
$y$ is interlaced by $x$.
\label{main2}
\end{theorem}

\begin{proof}

The proof is quite similar to the case
$L_{x<y}$, and we indicate only what needs
to be changed.

\noindent $\Longleftarrow$:  Without loss of generality we can assume
that $x$ is interlaced by $y$.
Using Lemma~\ref{count} we can build a finite
automaton recognizing $L_{x=y}$ just as we did
for $L_{x<y}$, using case (i).  The only difference
now is that the accepting state corresponds
to (b).

\bigskip
\noindent $\Longrightarrow$:  
Proceeding by contraposition, suppose that there is some $y$-bordered word $r$ such
that $x$ is not a subword of $r$, and some
$x$-bordered word $s$ such that $y$ is not a subword
of $s$.  Once again, we follow the argument
used for $L_{x<y}$, but there is one difference.

Recall that $z = (uv)^i (pq)^j $ for some $i>e$ and $j>f$.    By
the argument for $L_{x<y}$ we know that
$z$ has $(j-f)d'+c'-d'+m$ occurrences of $x$ and
$(i-e)d +c-d +n $ occurrences of $y$.  Let
$A = (-(m+c')) \bmod d'$ and $B = (-(n+c)) \bmod d$.
Let $w$ be the shortest suffix of
$(uv)^{e+2}$ such that $wz$ has 
$(i-e)d +c-d +n + B$ occurrences of $y$; let
$w'$ be the shortest prefix of $(pq)^{f+2}$ such
that $zw'$ has $(j-f)d'+c'-d'+m + A$ occurrences
of $x$.  Then by our construction 
$wzw'$ has $(j-f + C)d'$ occurrences of $x$ and
$(i-e+D)d$ occurrences of $y$, for some
$C, D \geq 0$.  

Now assume, contrary to what we want to prove,
that $L_{x=y}$ is regular.  Define $L' = L_{x=y} \cap w (uv)^e(uv)^+ (pq)^f(pq)^+ w'$.  Then $L'$ is regular.
Define $L = \# L' \#$, where $\#$ is a 
new symbol not in the alphabet $\Sigma$; then $L$ is regular.
Define a morphism $h: \{a,b,a',b'\}^* \rightarrow \Sigma^*$ as
follows:  $h(a') = \#w $, $h(a) = uv$,
$h(b) = pq$, and $h(b') = w'\#$.  We claim that
$h^{-1}( \# w z w' \#)= \{ a' a^i b^j b' \}$.  One direction is clear,
and the other follows from Theorem~\ref{factor}.
By a well-known theorem (e.g.,
\cite[Theorem 3.3.9]{Shallit:2009}), $h^{-1}(L)$
is regular.  But $h^{-1}(L) = 
\{ a' a^i b^j b' \ :\ (i-e+D)d = (j-f+C)d', \ \text{for} \ i>e, \ j>f \}$ which,
using the pumping lemma, is not regular.  
\end{proof}

\section{Testing the criteria}

Given $x,y$ we can test if there is some $y$-bordered word $z$ such that $x$ is not a subword of $z$, as follows: create a 
DFA recognizing the language 
$$(\Sigma^* x \Sigma^*)^c \ \cap \ y \Sigma^+  \ \cap \ \Sigma^+ y.$$
A simple construction gives such a DFA $M$ with
at most $N = (|x|+1)(|y|+3)(|y|+2)$ states and at most
$N |\Sigma|$ transitions.  

This can be
improved to $N' = (|x|+1)(2|y|+3)$ states as follows:
first build a DFA of $(2|y|+3)$ states
recognizing the language $y \Sigma^+  \ \cap \ \Sigma^+ y$ by ``grafting'' the
DFA, $A_1$, of $|y|+3$ states recognizing 
$y \Sigma^+ $ onto the DFA, $A_2$, of $|y|+2$ states
recognizing $\Sigma^+ y$.  This can be done
by modifying the pattern-matching DFA described
in Theorem~\ref{pma}. Simply replace transitions to the final state in $A_1$ with transitions to the appropriate states in $A_2$. The final state of $A_1$ and the initial state of $A_2$ both become unreachable.  Then form the direct product
with the DFA for $(\Sigma^* x \Sigma^*)^c$.
The resulting DFA
has $N'$ states.
We can then use a depth-first 
search on the underlying transition graph of $M$ to check if $L(M) \ne \emptyset$. 

Thus, we have proved:

\begin{corollary}
There is an algorithm running in time $O(|\Sigma||x||y|)$ that decides whether the criteria
of Theorems~\ref{main1} and \ref{main2} hold.
\end{corollary}

\begin{corollary}
If there exists a $y$-bordered word $z$ such that $x$ is not a subword of $z$, then
$|z| < N'$.
\label{four}
\end{corollary}

\begin{proof}
If $M = (Q,\Sigma, \delta, q_0, F)$ accepts any word at all, then it accepts a
word of length at most $|Q|-1$.   
\end{proof}

\section{Improving the bound in Corollary~\ref{four}}

As we have seen in Corollary~\ref{four}, if $x$ is not a subword of some
$y$-bordered word, then there is a relatively short
``witness'' to this fact.
We now show that this witness can be taken to be of the form $yty$ for some $t$ of \emph{constant length}. The precise constant depends on the cardinality of
the underlying alphabet $\Sigma$.  In Corollary~\ref{c3} we prove 
that if $|\Sigma| \geq 3$, then this constant is
$1$.  In Corollary~\ref{c2} we prove that if
$|\Sigma| = 2$, then this constant is $3$.

\begin{theorem}
Suppose $\Sigma$ is an alphabet that contains at least three symbols,
and let $x, y \in \Sigma^*$.  Without loss of generality
assume that $\{ 0, 1, 2 \} \subseteq \Sigma$.   If $x$ is a subword of
$y0y$ and $y1y$ and $y2y$, then $x$ is a subword of $y$.
\label{t5}
\end{theorem}

\begin{proof}
Assume, contrary to what we want to prove, that $x$ is not a subword of $y$.
Also assume that $|y| = m$ and $|x|= n$.
For $x$ to be a subword of $y0y$ (resp., $y1y$, $y2y$), then, it must be
that $x$ ``straddles'' the $y$---$y$ boundary.  More precisely, when we consider
where $x$ appears inside $y0y$, the first
symbol of $x$
must occur at or to the left of position $m+1$ of $y0y$ (resp., $y1y$, $y2y$).  Similarly, the last
symbol of $x$ must occur at or to the right of position $m+1$ of $y0y$ (resp., $y1y$,
$y2y$).  

For $a = 0, 1, 2$, label the $x$ that matches $yay$ as $x_a$, and assume that the position of
the $0$ that matches $x_0$ is $i$, the position of the $1$ that matches
$x_1$ is $j$, and the position of the $2$ that matches $x_2$ is $k$.
Note that $x_0 = x_1 = x_2 = x$; the indices just allow us to refer to the diagram
below.
Without loss of generality we can assume $1 \leq i < j < k \leq n$.
Thus we obtain a picture as in Figure~\ref{fig1}.  Here we have labeled
the two occurrences of $y$ as $y$ and $y'$, so we can refer
to them unambiguously.
Note that $i \geq 1$ and $k \leq m+1$.  Furthermore, note that 
$n \leq m+i$.

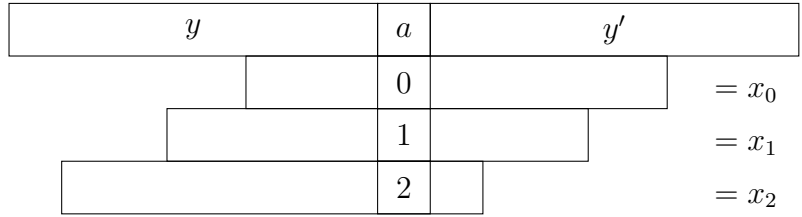
\begin{figure}[H]
\begin{center}
\begin{tikzpicture}[scale=0.7]
  \draw[black] (0,6) rectangle (7,5) node[pos=.5] {$y$};
  \draw[black] (7,6) rectangle (8,5) node[pos=.5] {$a$};
  \draw[black] (8,6) rectangle (15,5) node[pos=.5] {$y'$};
  
  \draw[black] (4.5,5) rectangle (12.5,4);
  \draw[black] (7,5) rectangle (8,4) node[pos=.5] {$0$};
  
  \draw[black] (3,4) rectangle (11,3);
  \draw[black] (7,4) rectangle (8,3) node[pos=.5] {$1$};
  
  \draw[black] (1,3) rectangle (9,2);
  \draw[black] (7,3) rectangle (8,2) node[pos=.5] {$2$};
  
  \node at (14, 4.3)   (x0) {$=x_0$};
  \node at (14, 3.3)   (x1) {$=x_1$};
  \node at (14, 2.3)   (x2) {$=x_2$};
  \end{tikzpicture}
\caption{Matches of $x$ against $y0y$, $y1y$, and $y2y$}
\label{fig1}
\end{center}
\end{figure}

We now use ``index-chasing'' to show that $x[k] = x[i]$; this will give us a
contradiction, since $x[k]= 2$ and $x[i] = 0$.
We will use the following identities, which can be deduced by observing Figure~\ref{fig1}.
\begin{align}
% x_0[\ell] &= y[\ell+m+1-i] \text{ for } 1 \leq \ell \leq i-1 \label{l1}\\
x_1[\ell] &= y[\ell+m+1-j] \text{  for } 1 \leq \ell \leq j-1; \label{l2}\\
x_2[\ell] &= y[\ell+m+1-k] \text{  for } 1 \leq \ell \leq k-1; \label{l3}\\
x_0[\ell] &= y'[\ell-i] \text{  for } i+1 \leq \ell \leq n; \label{l4}\\
x_1[\ell] &= y'[\ell-j] \text{  for } j+1 \leq \ell \leq n. \label{l5} 
% x_2[\ell] &= y'[\ell-k] \text{ for } k+1 \leq \ell \leq r \label{l6}.
\end{align}

Notice that $j+1 \leq k \leq n$, so we can
take $\ell = k$ in \eqref{l5} to get
$x[k] = y[k-j]$. Additionally, $k - j \geq 1$,
giving $i+1 \leq i+k-j$.  Also $i-j < 0 \leq n-k$, so $i+k-j \leq n$. Thus we can take $\ell = i+k-j$ in \eqref{l4} to obtain $y[k-j] = x[i+k-j]$.

Since $i \geq 1$ and $k -j \geq 1$,
we get $i + k - j \geq 2$.   Since $j-i \geq 1$ we have $i-j \leq -1$ and 
$i+k-j \leq k-1$.  Thus we can take $\ell = i+k-j$ in \eqref{l3} to get
$x[i+k-j] = y[i+m+1-j]$. Since $1 \leq i \leq j-1$, we 
we can take $\ell = i$ in \eqref{l2} to get $y[i+m+1-j] = x[i]$. Putting these observations together, we finally obtain 
$$2 = x[k] = y[k-j] = x[i+k-j] = y[i+m+1-j] = x[i] =0 ,$$
which produces the desired contradiction. 
\end{proof}

\begin{corollary}
Suppose $|\Sigma| \geq 3$.  Then $x$ is a subword of $yty$ for all $t$ with
$|t|= 1$ if and only if $y$ is interlaced by $x$.
\label{c3}
\end{corollary}

We now turn to case of a binary alphabet.  
This case is more subtle.
For example, consider when $x = 10100$ and
$y = 01001010$.  Then, as can be verified, $x$ is a
subword of the self-overlaps
$y (010)^{-1} y$ and $ y 0^{-1} y$, as well as
the words $yy$, $y0y$, $y1y$, $y00y$, $y01y$, $y10y$,
$y11y$, $y000y$, $y001y$, $y010y$, $y011y$, $y100y$,
$y101y$.  But $x$ is not a subword of $y110y$.

For a binary alphabet $\Sigma$, a special role is
played by the language 
$$A = 01^+ \, \cup \, 10^+ \, \cup \, 0^+1 \, \cup \, 1^+0 .$$

We also define the following languages. For each integer $k\geq 1$, let $B_{0^k1} :=(1+01+\cdots +0^{k-1}1)^+0^k0^*$ and $B_{10^k} :=0^*0^k(1+10+\cdots + 10^{k-1})^+$. Similarly, define $B_{1^k0}$ and $ B_{01^k}$ by relabeling 0 to 1 and 1 to 0.
\begin{lemma}
    Suppose  $\Sigma=\{0,1\}$ and $x\in A$. Then $y \in B_x$ if and only if $x$ is not a subword of $y$, but $x$ is a subword of all $y$-bordered words.
\label{lfour}
\end{lemma}
\begin{proof}
    We consider the case where $x=0^k1$ and note that the case where $x=1^k0$ is given by relabeling 0 to 1 and 1 to 0, and the other two cases are given by a symmetric argument.
    
    $\Longrightarrow$: Suppose that $y\in B_x=B_{0^k1}=(1+01+\cdots +0^{k-1}1)^+0^k0^*$ and $z$ is a $y$-bordered word. By the definition of $B_x$, observe that $x$ is not a subword of $y$. By Lemma \ref{ls} there exist $u\in\Sigma^+$, and $v\in\Sigma^*$,  and a natural number $e\geq 0 $ such that $y=(uv)^eu$ and $z=(uv)^{e+1}u$. 
    
    We first show that $e\leq 1$. If we assume the contrary, then $y=(uv)^{e-2}uvuvu$. We know that $vu$ has the suffix $10^{k+i}$ for some $i\geq 0$. But since there is at least one $1$ in $vu$ we have that $0^k1$ is a subword of $vuvu$, giving a contradiction. 
    
    If $e=0$ then $z=uvu=yvy$ for some $v\in\Sigma^*$. Since $vy$ has at least one $1$ and $y$ has a suffix of $0^k$, we get that $x$ is a subword of $yvy=z$. If $e=1$ then $y=uvu$ such that $vu$ has the suffix $0^k$ and there is at least one $1$ in $vu$. Then $z=(uv)^2u=uvuvu$ has $0^k1$ as a subword.
    
    $\Longleftarrow$: Assume, to get a contradiction, that there is some $y\in\Sigma^*\setminus B_x=\Sigma^*\setminus B_{0^k1}$ such that $x$ is a subword of all $y$-bordered words and $x$ is not a subword of $y$. Then $y$ satisfies at least one of the following cases, and we will get a contradiction in each of these. 
    
    {\bf Case (i):} $y=0^i$ for some $i\geq 0$. Clearly, $x$ is not a subword of the $y$-bordered word $yy$.
    
    {\bf Case (ii):} $y$ has $x=0^k1$ as a subword, giving an immediate contradiction.
    
    {\bf Case (iii):} The suffix of $y$ is $10^i$ for some $0\leq i<k$. Then consider the $y$-bordered word $z=y1y$. If $x$ is a subword of $z$ but $x$ is not a subword of $y$, then $x$ must straddle the $y$---$y$ boundary in $z$. So the $1$ in $x=0^k1$ must align with the $1$ between the $y$'s in $z=y1y$. But the suffix of $y$ is $10^i$ for $i<k$. So $x$  cannot be a subword of $y1y$. 
\end{proof}

However, for $x \not\in A$, it turns out that
if $x$ is not a subword of $y$, then
there is some word $t$ of length $3$ such that
$x$ is not a subword of $yty$.
To prove this we first give two preliminary lemmas.

\begin{lemma}\label{lem:difference1}
    Suppose $\Sigma=\{0,1\}$, and let $x,y\in\Sigma^*$ with $|x|=n$ and $|y|=m$. Suppose $x$ is not a subword of $y$, but $x$ is a subword of $yty$ for all $t\in\Sigma^*$ such that $|t|=3$ and $x\notin A$.  Then for every integer $k$ satisfying $\max\{1,m-n+2\} \leq k \leq \min\{2m+3-n,m+2\}$ and for all pairs of words $t_1,t_2\in\Sigma^*$ with $|t_1|=|t_2|=3$, we have either $x\neq(yt_1y)[k..k+n-1]$ or $x\neq(yt_2y)[k+1..k+n]$, or both.
    \label{five}
\end{lemma}
\begin{proof}
    Assume, to get a contradiction, that there exist $x,y\in\Sigma^*$ such that $x$ is not a subword of $y$ and $x\notin A$ and that there exist $t_1,t_2\in\Sigma^*$ with $|t_1|=|t_2|=3$ and an integer $k$ satisfying $\max\{1,m-n+2\}\leq k\leq \min\{2m+3-n,m+2\}$ such that $(yt_1y)[k..k+n-1]=(yt_2y)[k+1..k+n]=x$, and furthermore $x$ is a subword of $yty$ for all $t\in\Sigma^*$ with $|t|=3$. Let $t_1=a_1b_1c_1$ and $t_2=a_2b_2c_2$ and $x=x_1x_2\cdots x_n$. Before proceeding, first observe that $n\geq 3$ since for all $x$ with $|x|\leq 2$ we have that either $x\in A$ or $x$ is not a subword of one of $y000y$ and $y111y$.
\bigskip

\noindent{\bf Case (i):} $\max\{1,m-n+3\}\leq k\leq \min\{2m+3-n,m+1\}$.\\
    If $\max\{1,m-n+4\}\leq k\leq \min\{2m+3-n,m\}$ then $n\geq 4$ and we can write $x=x_1va_1b_1c_1w=va_2b_2c_2wx_n$ for some $v,w\in\Sigma^*$ where $x_1v=va_2$ and $c_1w=wx_n$ and $a_1b_1=b_2c_2$. Then, by the first theorem of Lyndon-Sch\"{u}tzenberger, we have that $v=x_1^i$ and $w=x_n^j$ for integers $i,j\geq 0$. Thus $x$ can be re-written as $x=x_1^ia_1b_1x_n^j$ for $x_1,a_1,b_1,x_n\in\Sigma$ and $i,j\geq 1$.
    
    If $k=m-n+3$ then, where $n\geq 3$, we can write $x=x_1va_1b_1=va_2b_2c_2$ for $v\in\Sigma^*$ and after applying the first theorem of Lyndon-Sch\"{u}tzenberger we get $x=x_1^ia_1b_1$ where $i\geq 1$.
    
    Similarly if $k=m+1$ then we can write $x=a_1b_1c_1w=b_2c_2wx_n$ and applying the first theorem of Lyndon-Sch\"{u}tzenberger gives $x=a_1b_1x_n^j$ for $j\geq 1$.
    
    So we have that $x=x_1^ia_1b_1x_n^j$ for $a_1,b_1,x_1,x_n\in\Sigma$ and $i,j\geq 0$ and either $i\geq 1$ or $j\geq 1$. We will proceed by getting a contradiction for each possible assignment of $a_1,b_1,x_1,x_n$ to symbols in $\Sigma$ for all valid $i,j$. Table \ref{tab:contradictions1} gives contradictions for all possible assignments where $x_1=0$. Note that the remaining cases can be ruled out by relabeling $0$ to $1$ and $1$ to $0$.
    
\bigskip
    
\noindent    {\bf Case (iii):} $k=m-n+2\geq 1$.\\
    We can write $x=x_1wa_1=wa_2b_2$ for some $w\in\Sigma^+$, where $x_1w=wa_2$. So by the first theorem of Lyndon-Sch\"{u}tzenberger, we get $w=x_1^i$ for some integer $i\geq 1$ and thus $x=x_1^{i+1}x_n$ for $x_1,x_n\in\Sigma$. If $x_1=x_n$, then $x_1^{i+1}x_n=x_1^{i+2}$. But, since $x$ is not a subword of $y$ we cannot have that $x_1^{i+2}$ is a subword of both $y111y$ and $y000y$, giving a contradiction. If instead we have $x_1\neq x_n$, then $x_1^{i+1}x_n\in A$, an immediate contradiction.
    
\bigskip
     
 \noindent   {\bf Case (iv):} $k=m+2\leq 2m+3-n$.\\
    We can write $x=b_1c_1w=c_2wx_n$ and similar to Case (iii), we get $x=x_1x_n^{i+1}$ for $x_1,x_n\in\Sigma$ and $i\geq 1$. By the same argument as in Case (iii), we get a contradiction if $x_1=x_n$ and if $x_1\neq x_n$.
   
    \begin{table}[H]
    	\caption{Contradictions for each $a_1,b_1,x_n\in\Sigma$ and $x_1=0$, where in each row $x=x_1^ia_1b_1x_n^j$ and $i,j\geq 0$ and either $i\geq 1$ or $j\geq 1$. The contradictions rely on the assumption that $x$ is not a subword of $y$.}\label{tab:contradictions1}
        \centering
        \renewcommand{\arraystretch}{1.2}
        \begin{tabular}{c|c|c|c|l}
            $x_1$ & $a_1$ & $b_1$ & $x_n$ & Contradiction \\
            \hline
            $0$ & 0 & 0 & $0$ & For all $i,j\geq 0$ we have that $x$ is not a subword of $y111y$.\\
            \hline 
            $0$ & 0 & 0 & $1$ & For all $i\geq 0$:\\[-3pt]
            &&&& If $j=0$, then $x$ is not a subword of $y111y$;\\[-3pt]
            &&&& If $j=1$, then $x=0^{i+2}1\in A$;\\[-3pt]
            &&&& If $j>1$, then if $x$ is a subword of $y101y$, then $y$ has $0^{i+2}1^{j-1}$ as a \\[-3pt]
            &&&& suffix. But if $x$ is a subword of $y011y$, then $y$ has $0^{i+1}$ as a suffix.\\
            \hline 
            $0$ & 0 & 1 & $0$ & For all $i\geq 0$:\\[-3pt]
            &&&& If $j=0$, then $x=0^{i+1}1\in A$;\\[-3pt]
            &&&& If $j>0$, then $x$ is not a subword of $y111y$.\\
            \hline 
            $0$ & 0 & 1 & $1$ & If $i=0$ or $j=0$, then $x\in A$.\\[-3pt]
            &&&& If $i>0$ and $j>0$, then if $x$ is a subword of $y101y$, then $y$ has $0^{i+1}1^j$ \\[-3pt]
            &&&& as a suffix. But if $x$ is a subword of $y011y$, then $y$ has $0^i$ as a suffix.\\
            \hline
            $0$ & 1 & 0 & $0$ & If $i=0$, then $x=10^{j+1}\in A$.\\[-3pt]
            &&&& If $i>0$, then $x$ is not a subword of $y111y$.\\
            \hline
            $0$ & 1 & 0 & $1$ & If $i=0$ and $j>0$, then $x$ is not a subword of $y000y$.\\[-3pt]
            &&&& If $i>0$ and $j=0$, then $x$ is not a subword of $y111y$.\\[-3pt]
            &&&& If $i>0$ and $j=1$, then if $x$ is a subword of $y011y$, then $y$ has $0^{i}1$\\[-3pt]
            &&&& as a suffix. But if $x$ is a subword of $y111y$, then $y$ has $0^i10$ as a suffix.\\[-3pt]
            &&&& If $i=1$ and $j>0$, then if $x$ is a subword of $y001y$, then $y$ has $01^j$\\[-3pt]
            &&&& as a prefix. But if $x$ is a subword of $y000y$, then $y$ has $101^j$ as a prefix.\\[-3pt]
            &&&& If $i>1$ and $j>1$, then if $x$ is a subword of $y011y$, then $y$ has $0^i1$\\[-3pt]
            &&&& as a suffix. But if $x$ is a subword of $y111y$, then $y$ has $0^i101^\ell$ as a suffix\\[-3pt]
            &&&& for some $\ell <j$. Since $i>1$, this is a contradiction.\\
            \hline
            $0$ & 1 & 1 & $0$ & If $i=0$ and $j=1$, then $x=110\in A$.\\[-3pt]
            &&&& If $i=1$ and $j=0$, then $x=011\in A$.\\[-3pt]
            &&&& If $i>1$ and $j=0$, then if $x$ is a subword of $y101y$, then $y$ has $0^i1$\\[-3pt]
            &&&& as a suffix. But if $x$ is a subword of $y011y$, then $y$ has $0^{i-1}$ as a suffix.\\[-3pt]
            &&&& If i=0 and $j>1$, then if $x$ is a subword of $y101y$, then $y$ has $10^j$\\[-3pt]
            &&&& as a prefix. But if $x$ is a subword of $y110y$, then $y$ has $0^{j-1}$ as a prefix.\\[-3pt]
            &&&& If $i>0$ and $j>0$, then $x$ is not a subword of $y111y$.\\
            \hline
            $0$ & 1 & 1 & $1$ & For all $j\geq 0:$\\[-3pt]
            &&&& If $i=0$, then $x$ is not a subword of $y000y$;\\[-3pt]
            &&&& If $i=1$, then $x=01^{j+2}\in A$;\\[-3pt]
            &&&& If $i>1$, then if $x$ is a subword of $y011y$, then $y$ has $0^{i-1}$ as a suffix.\\[-3pt]
            &&&& But if $x$ is a subword of $y101y$, then $y$ has $0^i1^{j+1}$ as a suffix.
        \end{tabular}
        
    \end{table}
\end{proof}

\begin{lemma}\label{lem:difference2}
  Suppose $\Sigma=\{0,1\}$, and let $x,y\in\Sigma^*$ with $|x|=n$ and $|y|=m$. If $x$ is a subword of $yty$ for all $t\in\Sigma^*$ such that $|t|=3$ and $x\notin A$, and $x$ is not a subword of $y$, then for all pairs of words $t_1,t_2\in\Sigma^*$ with $|t_1|=|t_2|=3$ we have either $x\neq (yt_1y)[m+1..m+n]$, or $x\neq(yt_2y)[m+3..m+2+n]$, or both.
  \label{six}
\end{lemma}
\begin{proof}
  Assume, to get a contradiction, that there exist $x,y\in\Sigma^*$ such that $x$ is not a subword of $y$ and $x\notin A$ and that there exist $t_1,t_2\in\Sigma^*$ with $|t_1|=|t_2|=3$ such that $(yt_1y)[m+1..m+n]=(yt_2y)[m+3..m+2+n]=x$, and furthermore $x$ is a subword of $yty$ for all $t\in\Sigma^*$ with $|t|=3$. Let $t_1=a_1b_1c_1$ and $t_2=a_2b_2c_2$ and $x=x_1x_2\cdots x_n$, and assume $|y|=m$. 
  
  We can write $x=a_1b_1c_1w=c_2wx_{n-1}x_n$. So $b_1c_1w=wx_{n-1}x_n$ and by the first theorem of Lyndon-Sch\"{u}tzenberger there exist $u\in\Sigma^+$ and $v\in\Sigma^*$ and an integer $i\geq 0$ such that $b_1c_1=uv$ and $x_{n-1}x_n=vu$ and $w=(uv)^iu=u(vu)^i$. This gives $x=x_1wx_{n-1}x_n=x_1(uv)^iuvu=x_1(uv)^{i+1}u$. We now consider each possible $u\in\Sigma^+$ and $v\in\Sigma^*$, seeking a contradiction in each case. The contradictions are summarized in Table \ref{tab:contradictions2}. Note again that the contradictions are given for all cases where $x_1=0$; the remaining cases can be obtained by relabeling 0 to 1 and 1 to 0.
  
  \begin{table}[H]
  	\caption{Contradictions for each valid $u\in\Sigma^+$, $v\in\Sigma^*$, and $x_1=0$, where in each row $x=x_1(uv)^{i+1}u$ for $i\geq 0$. The contradictions rely on the assumption that $x$ is not a subword of $y$.}\label{tab:contradictions2}
    \centering
    \renewcommand{\arraystretch}{1.2}
    \begin{tabular}{c|c|c|l}
      $x_1$ & $u$ & $v$ & Contradiction\\
      \hline
      0 & 00 & $\epsilon$ & $x$ is not a subword of $y111y$.\\
      \hline
      0 & 0 & 0 & $x$ is not a subword of $y111y$.\\
      \hline 
      0 & 01 & $\epsilon$ & If $x$ is a subword of $y111y$, then $y$ has $0(01)^{i+1}0$ as a suffix.\\[-3pt]
      &&& But if $x$ is a subword of $y011y$, then $y$ has $0(01)^{i+1}$ as a suffix.\\
      \hline
      0 & 0 & 1 & $x$ is not a subword of $y111y$.\\
      \hline
      0 & 10 & $\epsilon$ & $x$ is not a subword of $y111y$.\\
      \hline
      0 & 1 & 0 & If $x$ is a subword of $y111y$ then $y$ has $0(10)^{i+1}$ as a suffix.\\[-3pt]
      &&& But if $x$ is a subword of $y011y$ then $y$ has $(01)^{i+1}$ as a suffix.\\
      \hline
      0 & 11 & $\epsilon$ & $x\in A$.\\
      \hline
      0 & 1 & 1 & $x\in A$.\\
    \end{tabular}
  \end{table}
\end{proof}

\begin{theorem}\label{thm:ytyNotInA}
    Suppose $\Sigma = \{ 0,1 \}$ and let $x, y \in \Sigma^*$.   If $x$ is a subword of $yty$ for all $t\in\Sigma^*$ such that $|t|=3$ and $x\notin A$, then $x$ is a subword of $y$.
\end{theorem}
\begin{proof}
    Define the function $f:\Sigma^*\times \Sigma^*\rightarrow \mathbb{N}$ over pairs of words $x,w\in\Sigma^*$ such that $x$ is a subword of $w$ as $f(x,w)=\min\{i\in\mathbb{N}:w[i..i+|x|-1]=x\}$. Also, define the bitwise complements of elements of $\Sigma$ as $\overline{0}=1$ and $\overline{1}=0$.
    
    Assume, to get a contradiction, that $x$ is not a subword of $y$ and also assume that $|y|=m$ and $|x|=n$. If $x$ is a subword of $yty$ for each $t\in\Sigma^*$ with $|t|=3$, then for each such $t$  we have $f(x,yty)\leq m+3$ and $f(x,yty)+n-1\geq m+1$. Since the position of $x$ in $yty$ cannot be the same for all valid $t$, let $t_0=a_0b_0c_0$ be the choice of $t$ for which $f(x,yty)$ is greatest across all valid $t$ and let $t_4=a_4b_4c_4\neq t_0$ be the choice of $t$ for which $f(x,yty)$ is smallest across all valid $t$. We now consider two cases depending on the position of $x$ in $yt_0y$.

\bigskip
    
\noindent    {\bf Case (i):} $f(x,yt_0y)=m+3$.
Consider $t_1=\overline{a_4} 0 \overline{c_0}$ and $t_2=\overline{a_4} 1 \overline{c_0}$. Since $t_1$ and $t_2$ differ from $t_4$ in the first index of $t_4$, and $t_4$ gives the leftmost position for $x$ as a subword of $yty$ over all valid choices of $t$, we know $f(x,yt_1y)\neq f(x,yt_4y)$ and $f(x,yt_2y)\neq f(x,yt_4y)$. Similarly we have $f(x,yt_1y)\neq f(x,yt_0y)$ and $f(x,yt_2y)\neq f(x,yt_0y)$. Applying Lemma \ref{lem:difference1} to the pairs $t_4,t_1$ and $t_4,t_2$ and Lemmas \ref{lem:difference1} and \ref{lem:difference2} to the pairs $t_1,t_0$ and $t_2,t_0$ we have that $f(x,yt_1y)+n-1\geq m+3$ and $f(x,yt_2y)+n-1\geq m+3$ and that $f(x,yt_1y)\leq m+1$ and $f(x,yt_2y)\leq m+1$. So for $yt_1y$ (resp., $yt_2y$), the position of $x$ is such that it entirely overlaps $t_1$ (resp., $t_2$). But since $t_1\neq t_2$ we know that the positions of $x$ as a subword of $yt_1y$ and $yt_2y$ are distinct, i.e., $f(x, yt_1y)\neq f(x,yt_2y)$. 

So suppose without loss of generality that $f(x,yt_1y)>f(x,yt_2y)$. We now perform an index chasing argument, similar to that of the ternary case, using $t_0,t_1,t_2$ and seeking the contradiction $c_0=(yt_0y)[m+3]=(yt_2y)[m+3]=\overline{c_0}$. We use the same labeling scheme as in the ternary case. So define $i,j,k$ such that $i=m+4-f(x,yt_0y)$ and $j=m+4-f(x,yt_1y)$ and $k=m+4-f(x,yt_2y)$, giving $x_0[i]=t_0[3]=c_0$ and $x_1[j]=t_1[3]=\overline{c_0}$ and $x_2[k]=t_2[3]=\overline{c_0}$. Note that in this case we have $i=1$ by assumption and  $j\geq i+3$ by Lemmas \ref{lem:difference1} and \ref{lem:difference2}. From Figure \ref{fig:binaryCase1} we obtain the following identities.
\begin{align}
    x_1[\ell] &= y[\ell+m+3-j] \text{  for } 1 \leq \ell \leq j-3; \label{b2}\\
    x_2[\ell] &= y[\ell+m+3-k] \text{  for } 1 \leq \ell \leq k-3; \label{b3}\\
    x_0[\ell] &= y'[\ell-i] \text{  for } i+1 \leq \ell \leq n; \label{b4}\\
    x_1[\ell] &= y'[\ell-j] \text{  for } j+1 \leq \ell \leq n. \label{b5}
    \end{align}
    Since $j+1\leq k\leq n$, we can apply \eqref{b5} to get $x_1[k]=y'[k-j]$. Then, since $i\leq j$ and $k\leq n$, we have $i+k\leq n+j$, giving $i+k-j\leq n$. This, together with the inequality $k-j\geq 1$ giving $i+1\leq i+k-j$, means that we can apply \eqref{b4} to get $y'[k-j]=x_0[i+k-j]$. Next, $j-i\geq 3$ gives $i+k-j\leq k-3$, so applying \eqref{b3} gives $x_2[k+i-j]=y[i+m+3-j]$. Finally, since $1\leq i\leq j-3$,  we can apply \eqref{b2} to get $y[i+m+3-j]=x_1[i]$. Together this gives the contradiction
    $$\overline{c_0}=x_2[k]=y'[k-j]=x_0[i+k-j]=y[i+m+3-j]=x_2[i]=c_0.$$
    \begin{figure}
        \centering
        \begin{tikzpicture}[scale=0.8]
        \draw[black] (0,6) rectangle (8,5) node[pos=.5] {$y$};
        \draw[black] (8,6) rectangle (9,5) node[pos=.5] {$a$};
        \draw[black] (9,6) rectangle (10,5) node[pos=.5] {$b$};
        \draw[black] (10,6) rectangle (11,5) node[pos=.5] {$c$};
        \draw[black] (11,6) rectangle (19,5) node[pos=.5] {$y'$};
        
        \draw[black] (10,5) rectangle (17,4);
        \draw[black] (10,5) rectangle (11,4) node[pos=.5] {$c_0$};
        
        \draw[black] (7,4) rectangle (14,3);
        \draw[black] (8,4) rectangle (9,3) node[pos=.5] {$\overline{a_4}$};
        \draw[black] (9,4) rectangle (10,3) node[pos=.5] {$0$};
        \draw[black] (10,4) rectangle (11,3) node[pos=.5] {$\overline{c_0}$};
        
        \draw[black] (5,3) rectangle (12,2);
        \draw[black] (8,3) rectangle (9,2) node[pos=.5] {$\overline{a_4}$};
        \draw[black] (9,3) rectangle (10,2) node[pos=.5] {$1$};
        \draw[black] (10,3) rectangle (11,2) node[pos=.5] {$\overline{c_0}$};
        
        \node at (18, 4.3)   (x0) {$=x_0$};
        \node at (18, 3.3)   (x1) {$=x_1$};
        \node at (18, 2.3)   (x2) {$=x_2$};
        \end{tikzpicture}
        \caption{Positions of $x$ in $yt_0y,yt_1y,yt_2y$ for Case (i).}\label{fig:binaryCase1}
    \end{figure}
    
 \noindent   {\bf Case (ii):} $f(x,yt_0y)\leq m+2$.
 \smallskip\noindent
    Consider $t_1=\overline{a_4 b_0}c_0$ and $t_2=\overline{a_4}b_0\overline{c_0}$ and $t_3=\overline{a_4 b_0 c_0}$. By the same argument as in Case 1 we get that the positions of $x$ as a subword of each of $yt_0y,yt_1y,yt_2y,yt_3y,yt_4y$ are all distinct. Furthermore, by Lemma \ref{lem:difference1} we have that for each pair $t_i,t_j$ with $0\leq i,j\leq 4$ and $i\neq j$ the difference in positions of $x$ as a subword of $yt_iy$ and $yt_jy$ is $|f(x,yt_iy)-f(x,yt_jy)|\geq 2$. We now order these choices of $t$ and relabel $t_1,t_2,t_3$ if necessary such that $f(x,yt_0y)> f(x,yt_3y)> f(x,yt_1y)>f(x,yt_2y)> f(x,yt_4y)$. At this point we again perform an index-chasing argument using $t_0,t_1,t_2$. If we have that $t_2[3]=\overline{c_0}$ then the argument given in Case (i) holds to give a contradiction. If instead we have that $t_2[3]=c_0$, then we know that $t_2[2]=\overline{b_0}$ and we will get the contradiction $b_0=(yt_0y)[m+2]=(yt_2y)[m+2]=\overline{b_0}$.  To do this we define $i,j,k$ such that $i=m+3-f(x,yt_0y)$ and $j=m+3-f(x,yt_1y)$ and $k=m+3-f(x,yt_2y)$, giving $x_0[i]=t_0[2]=b_0$ and $x_1[j]=t_1[2]$ and $x_2[k]=t_2[2]=\overline{b_0}$. Since $f(x,yt_0y)>f(x,yt_3y)>f(x,yt_1y)$, Lemma \ref{lem:difference1} gives $f(x,yt_0y)\geq f(x,yt_3y)+2$ and $f(x,yt_3y)\geq f(x,yt_1y)+2$. So $f(x,yt_0y)\geq f(x,yt_1y) + 4$ and thus $j\geq i+4$. From Figure \ref{fig:binaryCase2} we get the following identities. Note that these identities are centered around $t[2]$ instead of $t[3]$ as in Case 1.
    \begin{align}
    x_1[\ell] &= y[\ell+m+2-j] \text{  for } 1 \leq \ell \leq j-2; \label{b7}\\
    x_2[\ell] &= y[\ell+m+2-k] \text{  for } 1 \leq \ell \leq k-2; \label{b8}\\
    x_0[\ell] &= y'[\ell-i-1] \text{  for } i+2 \leq \ell \leq n; \label{b9}\\
    x_1[\ell] &= y'[\ell-j-1] \text{  for } j+2 \leq \ell \leq n. \label{b10}
    \end{align}
    Since $j+2\leq k\leq n$ we can apply \eqref{b10} to get $x_1[k]=y'[k-j-1]$. Then since $i\leq j$ and $k\leq n$ we have $i+k\leq n+j$, giving $i+k-j-1<i+k-j\leq n$. This, together with $k-j\geq 2$ giving $i+2\leq i+k-j$ by Lemma \ref{lem:difference1}, means that we can apply \eqref{b9} to get $y'[k-j-1]=x_0[i+k-j]$. Next, $j-i\geq 4$ gives $i+k-j\leq k-4<k-2$, so applying \eqref{b8} gives $x_2[k+i-j]=y[i+m+2-j]$. Finally, since $1\leq i\leq j-4<j-2$,  we can apply \eqref{b7} to get $y[i+m+2-j]=x_1[i]$. Together this gives the contradiction
$$\overline{b_0}=x_2[k]=y'[k-j-1]=x_0[i+k-j]=y[i+m+2-j]=x_2[i]=b_0.$$

    \begin{figure}[H]
        \centering
        \begin{tikzpicture}[scale=0.8]
        \draw[black] (0,6) rectangle (8,5) node[pos=.5] {$y$};
        \draw[black] (8,6) rectangle (9,5) node[pos=.5] {$a$};
        \draw[black] (9,6) rectangle (10,5) node[pos=.5] {$b$};
        \draw[black] (10,6) rectangle (11,5) node[pos=.5] {$c$};
        \draw[black] (11,6) rectangle (19,5) node[pos=.5] {$y'$};
        
        \draw[black] (8.5,5) rectangle (16.5,4);
        \draw[black] (9,5) rectangle (10,4) node[pos=.5] {$b_0$};
        \draw[black] (10,5) rectangle (11,4) node[pos=.5] {$c_0$};
        
        \draw[black] (6,4) rectangle (14,3);
        \draw[black] (8,4) rectangle (9,3) node[pos=.5] {$\overline{a_4}$};
        \draw[black] (9,4) rectangle (10,3) node[pos=.5] {\footnotesize $t_1[2]$};
        \draw[black] (10,4) rectangle (11,3) node[pos=.5] {$\overline{c_0}$};
        
        \draw[black] (4,3) rectangle (12,2);
        \draw[black] (8,3) rectangle (9,2) node[pos=.5] {$\overline{a_4}$};
        \draw[black] (9,3) rectangle (10,2) node[pos=.5] {$\overline{b_0}$};
        \draw[black] (10,3) rectangle (11,2) node[pos=.5] {$c_0$};
        
        \node at (18, 4.3)   (x0) {$=x_0$};
        \node at (18, 3.3)   (x1) {$=x_1$};
        \node at (18, 2.3)   (x2) {$=x_2$};
        \end{tikzpicture}
        \caption{Positions of $x$ in $yt_0y,yt_1y,yt_2y$ for Case (ii) where $t_2[3]=c_0$.}\label{fig:binaryCase2}
    \end{figure}
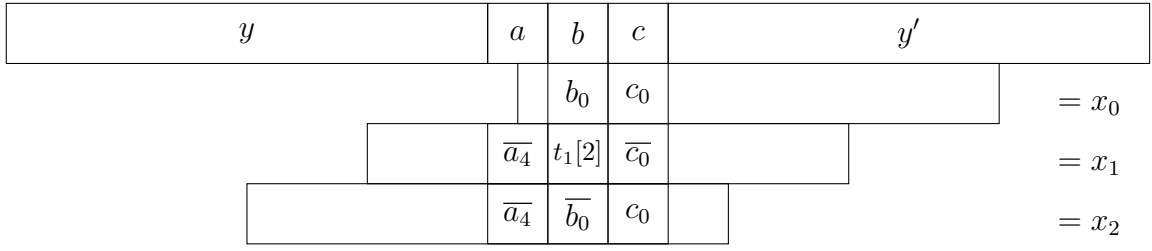
\end{proof}

\begin{corollary}
  Let $\Sigma = \{0,1\}$.  Then $x$ is a subword of every $y$-bordered 
  word if and only if $x$ is a subword of $yty$ for all words $t$ of length 3.
\label{c2}
\end{corollary}  

\begin{proof}
  If $x$ is a subword of every $y$-bordered word, then clearly $x$ is a subword of $yty$ for all words $t$ of length 3.
  For the other direction there are two cases. 
\medskip

\noindent {\bf Case 1:} $x\notin A$. Then by Theorem \ref{thm:ytyNotInA} we know $x$ is a subword of $y$. So $x$ is also a subword of every $y$-bordered word.

\medskip
  
\noindent{\bf Case 2:} $x\in A$. Then $x$ has the form $01^i$, or $0^i1$, or $10^i$, or $1^i0$ for some $i\geq 1$. We consider the case where $x=01^i$ and note that the case where $x=1^i0$ follows by a symmetric argument and the other cases are given by relabeling 0 to 1 and 1 to 0. If $x$ is a subword of $y$ then the result follows trivially. So suppose that $x=01^i$ is not a subword of $y$, but that $x$ is a subword of $yty$ for all words $t$ of length 3. Then, since $x$ is a subword of $y000y$, we have that $1^i$ is a prefix of $y$. Additionally, since $x$ is a subword of $y111y$, we know that $01^j$ is a suffix of $y$ for some $j$ satisfying $0\leq j<i$. So we have $y=1^iw01^j$ for some $w\in\Sigma^*$. Now consider a $y$-bordered word $z$. Let $k$ be the index of the first 0 in $z$. Since $z$ has $y$ as a prefix and a suffix, and $z\neq y$, we know that $|z|\geq |y|+k$. This is because the $y$-suffix of $z$ must start after the first $0$ in $z$. So we have that there are $i$ consecutive $1$'s in $z$ starting at some index $\ell>k$. Let $k'$ be the largest index less than $\ell$ such that $z[k']=0$. Then $z[k'..k'+i]=01^i$. So $x$ is a subword of $z$. 
\end{proof}

\begin{remark}
The number $3$ is optimal in Corollary~\ref{c2}.
Consider $x = 10100$, $y = 01001010$.   
Then $x$ is a subword of every $y$-bordered
word of length $\leq 2|y| + 2 = 18$, but not
a subword of $yty$ with $t = 110$.
\end{remark}

\section{Finiteness}

We now examine when $L_{x=y}$ is finite. 

\begin{theorem}
Let $x, y \in \Sigma^+$.  Then $L_{x=y}$ is finite if and only if
$|\Sigma|= 1$ and $x \not= y$.
\label{seven}
\end{theorem}

\begin{proof}
There are four cases to consider.

\noindent {\bf Case (i):}  $x = a^i$ and $y = a^j$ for integers $i, j > 0$.
If $|\Sigma| = 1$, then $L_{x=y}$ is finite if and only if $x \not= y$, for
otherwise without loss of generality $i < j$, and for $n \geq j$ the word
$a^n$ contains $n-j+1$ occurrences of $a^j$, but $n-i+1$ occurrences of $a^i$.

Otherwise $|\Sigma| > 1$.  Let $b \in \Sigma$ and $b \not= a$.  Then
for each $z \in b^*$ we have $|z|_x = |z|_y = 0$.  Thus $L_{x=y}$ is infinite.

\noindent {\bf Case (ii):}  $x = a^i$ and $y = b^j$ for two distinct symbols $a, b$
and $i, j > 0$.  Then for each $z$ of the form $(xy)^n$ we have
$|z|_x = |z|_y = n$.   Thus $L_{x=y}$ is infinite.

\noindent {\bf Case (iii):}  $x = a^i$ for some $i > 0$ but $y$ contains two different symbols.  Let $b \in \Sigma$ with $b\not= a$.  Then
for each $z \in b^*$ we have $|z|_x = |z|_y = 0$.  Thus $L_{x=y}$ is infinite.

\noindent {\bf Case (iv):}  $x$ and $y$ both contain two different symbols.
Let $a \in \Sigma^*$.  Then for each $z \in a^*$ we have $|z|_x = |z|_y = 0$.  Thus $L_{x=y}$ is infinite. 
\end{proof}

We could consider the generalization of 
$L_{x=y}$ to more than two words:
$$L_{x_1=x_2=\cdots=x_n} = 
\{ z \in \Sigma^* \ : \ |z|_{x_1} = |z|_{x_2} = \cdots = |z|_{x_n} \}.$$ 
The following examples show that deciding the finiteness of
$L_{x_1=x_2=\cdots=x_n}$ for $n \geq 3$ is more subtle than the case
$n = 2$.
Suppose
$\Sigma = \{0,1\}$.  Then
   $L_{0 = 1 = 00 = 11}$ and 
    $L_{0 = 1 = 01 = 10}$ are finite languages, but
   $L_{00 = 11 = 000 = 111}$ is not.
   
Consider $L_{0 = 1 = 00 = 11}$.  For any maximal subword consisting of 0's, the number of 0's exceeds the number of 00's, and similar for 1 and 11. So $L_{0 = 1 = 00 = 11} = \{\epsilon \}$.

Consider $L_{0 = 1 = 01 = 10}$. Since $|z|_{01} = |z|_{10}$, as shown in Figure~\ref{a01}, the words in this language must start and end with the same character.
There cannot be a 00 or the number of 0's exceeds that of 01 and 10, and similar for 11.
So, the language is a subset of $(01)^* 0 \cup (10)^* 1 \cup \{\epsilon\}$.
But no word $z$ in this language, other than $\epsilon$, has $|z|_0 = |z|_1$.
Therefore, $L_{0 = 1 = 01 = 10} = \{\epsilon\}$.

Consider $L_{00 = 11 = 000 = 111}$. It contains $(01)^*$, and hence is infinite.

Lacking a general condition for finiteness, we prove the following sufficient condition.

\begin{theorem}
If $|x_1|=\cdots=|x_n|$ then
$L_{x_1=x_2=\cdots=x_n}$ is infinite.
\end{theorem}

\begin{proof}
Let $\ell = |x_1|$.  Consider the cyclic 
order-$\ell$ de Bruijn word $w$ of length
$k^\ell$ over the cardinality-$k$ alphabet 
$\Sigma$.  Such a word is guaranteed to exist
for all $k \geq 2$ and $\ell \geq 1$; see, e.g., 
\cite{Ralston:1982}.  Let $w'$ be the prefix of $w$ of
length $\ell-1$.  Then $w^i w' \in L_{x_1=x_2=\cdots=x_n}$ for all $i \geq 1$.
\end{proof}


\begin{thebibliography}{1}

\bibitem{Cormen&Leiserson&Rivest&Stein:2001}
T.~H. Cormen, C.~E. Leiserson, R.~L. Rivest, and C.~Stein.
\newblock {\em Introduction to Algorithms}.
\newblock MIT Press, 2nd edition, 2001.

\bibitem{Ehrenfeucht&Silberger:1979}
A.~Ehrenfeucht and D.~M. Silberger.
\newblock Periodicity and unbordered segments of words.
\newblock {\em Discrete Math.} {\bf 26} (1979), 101--109.

\bibitem{Gamard&Richomme&Shallit&Smith:2017}
G.~Gamard, G.~Richomme, J.~Shallit, and T.~J. Smith.
\newblock Periodicity in rectangular arrays.
\newblock {\em Info. Proc. Letters} {\bf 118} (2017), 58--63.

\bibitem{Hopcroft&Ullman:1979}
J.~E. Hopcroft and J.~D. Ullman.
\newblock {\em Introduction to Automata Theory, Languages, and Computation}.
\newblock Addison-Wesley, 1979.

\bibitem{Lothaire:1983}
M.~Lothaire.
\newblock {\em Combinatorics on Words}, Vol.~17 of {\em Encyclopedia of
  Mathematics and Its Applications}.
\newblock Addison-Wesley, 1983.

\bibitem{Lyndon&Schutzenberger:1962}
R.~C. Lyndon and M.~P. {Sch\"utzenberger}.
\newblock The equation {$a^M = b^N c^P$} in a free group.
\newblock {\em Michigan Math. J.} {\bf 9} (1962), 289--298.

\bibitem{Parikh:1966}
R.~J. Parikh.
\newblock On context-free languages.
\newblock {\em J. ACM} {\bf 13} (1966), 570--581.

\bibitem{Ralston:1982}
A.~Ralston.
\newblock {De Bruijn} sequences --- a model example of the interaction of
  discrete mathematics and computer science.
\newblock {\em Math. Mag.} {\bf 55} (1982), 131--143.

\bibitem{Shallit:2009}
J.~Shallit.
\newblock {\em A Second Course in Formal Languages and Automata Theory}.
\newblock Cambridge Univ. Press, 2009.

\end{thebibliography}
\end{document}